\documentclass[a4paper,11pt,onecolumn, accepted=2023-07-04]{quantumarticle}
\pdfoutput=1

\usepackage[margin=1in]{geometry}
\usepackage{palatino}
\usepackage{subcaption}
\usepackage[T1]{fontenc}

\usepackage{enumitem}
\usepackage{booktabs}
\usepackage{float}
\usepackage{makecell}

\usepackage[
  colorlinks,
  linkcolor = blue,
  citecolor = blue,
  urlcolor = blue]{hyperref}

\usepackage{amsmath,amssymb,amsthm}
\usepackage{mathtools}
\mathtoolsset{centercolon}

\usepackage[affil-it]{authblk}
\usepackage{tabu}
\usepackage{tikz}

\makeatletter
\newtheorem*{rep@theorem}{\rep@title}
\newcommand{\newreptheorem}[2]{%
\newenvironment{rep#1}[1]{%
 \def\rep@title{#2 \ref{##1}}%
 \begin{rep@theorem}}%
 {\end{rep@theorem}}}
\makeatother

\newreptheorem{theorem}{Theorem}
\newreptheorem{lemma}{Lemma}

\newtheorem{theorem}{Theorem}[section]
\newtheorem*{theorem*}{Theorem}
\newtheorem{lemma}[theorem]{Lemma}
\newtheorem{cor}[theorem]{Corollary}

\newtheorem*{theoremA}{Theorem A}
\newtheorem*{theoremB}{Theorem B}

\theoremstyle{definition}
\newtheorem{definition}[theorem]{Definition}
\newtheorem{remark}[theorem]{Remark}
\newtheorem{example}[theorem]{Example}

\newcommand{\ket}[1]{|#1\rangle}
\newcommand{\bra}[1]{\langle#1|}

\newcommand{\ct}{^{\dagger}}
\newcommand{\tp}{^{\mathsf{T}}}

\newcommand{\N}{\mathbb{N}}
\newcommand{\C}{\mathbb{C}}

\title{Counterexamples in self-testing}
\author{Laura Man\v{c}inska}\email{mancinska@math.ku.dk}
\author{Simon Schmidt}\email{sisc@math.ku.dk}
\affiliation{QMATH, Department of Mathematical Sciences, University of Copenhagen, Universitetsparken 5, 2100 \linebreak Copenhagen \O, Denmark}

\begin{document}

\maketitle

\begin{abstract}
In the recent years self-testing has grown into a rich and active area of study with applications ranging from practical verification of quantum devices to deep complexity theoretic results. Self-testing allows a classical verifier to deduce which quantum measurements and on what state are used, for example, by provers Alice and Bob in a nonlocal game. Hence, self-testing as well as its noise-tolerant cousin---robust self-testing---are desirable features for a nonlocal game to have. 

Contrary to what one might expect, we have a rather incomplete understanding of if and how self-testing could fail to hold. In particular, could it be that every 2-party nonlocal game or Bell inequality with a quantum advantage certifies the presence of a specific quantum state? Also, is it the case that every self-testing result can be turned robust with enough ingeniuty and effort? We answer these questions in the negative by providing simple and fully explicit counterexamples. To this end, given two 
 nonlocal games $\mathcal{G}_1$ and $\mathcal{G}_2$, we introduce the $(\mathcal{G}_1 \lor \mathcal{G}_2)$-game, in which the players get pairs of questions and choose which game they want to play. The players win if they choose the same game and win it with the answers they have given. Our counterexamples are based on this game and we believe this class of games to be of independent interest.
\end{abstract}

\section{Introduction}
The notion of self-testing was first introduced by Mayers and Yao \cite{MY} and it has since developed into an active and rich area of study (see \cite{SB} for a review). One major motivation behind self-testing is that it can be used by a classical verifier to certify that untrusted quantum devices perform according to their specification. This is accomplished by deriving a quantum mechanical description of a quantum device merely from classical observations.
More precisely, if Alice and Bob want to know the state and measurements of their quantum devices, they can play a nonlocal game and check the output probabilities. If this game is a self-test then after observing the desired output probabilities, Alice and Bob can conclude that their devices must be implementing certain measurements on a certain quantum state.

In addition to the early applications of self-testing to certification and device - independent protocols, in the recent years it has been a key ingredient for important results in quantum complexity theory \cite{Fitzsimons, NatarajanVidick, NatarajanWright, mipre}. The most notable among these results is the recent breakthrough \cite{mipre} establishing that MIP*=RE and resolving Connes' Embedding Problem which had resisted all attempts for over 50 years. 

It was shown by Bell \cite{Bell64} that the predictions of quantum mechanics are not compatible with any local hidden-variable theory. He obtained this result by showing that quantum mechanics predicts violations of what are now known as \emph{Bell inequalities}. Our definition of self-testing will be given in the setting of \emph{nonlocal games} \cite{CHTW} which are another way to explore Bell inequalities and their violations.
A (2-player) nonlocal game is played by two collaborating players, Alice and Bob, and a referee. The referee gives questions to Alice and Bob and they each have to respond with an answer. The players win according to a function on questions and answers, which is known to the players beforehand. Crucially, Alice and Bob are not allowed to communicate after they have received their questions. They can, however, agree on a strategy ahead of time. In a \emph{quantum strategy} Alice and Bob can  share an entangled state and use local measurements to come up with their answers. Finding a quantum strategy for a nonlocal game which exceeds the winning probability of any classical strategy yields a violation of a Bell inequality. We call quantum strategies with the highest winning probability in a nonlocal game \emph{optimal} quantum strategies. 

\paragraph{Motivation and results.} In a nutshell, \emph{self-testing} says that any optimal quantum strategy for some game $\mathcal G$ is equal to a reference strategy, up to local isometries. There are different forms of self-testing. A game can self-test an optimal quantum strategy or just the shared state of this quantum strategy. Also, a game can \emph{robustly} self-test a strategy in which case we in addition require that any near-optimal quantum strategy must be close to an optimal reference strategy. For a nonlocal game $\mathcal{G}$, an optimal quantum strategy $S$, and state $\ket{\psi}$ used in $S$,  we have the following relation between the above forms of self-testing:
\begin{align}
    \mathcal{G} \text{ robustly self-tests }S \Rightarrow \mathcal{G} \text{ self-tests }S \Rightarrow \mathcal{G} \text{  self-tests }\ket{\psi}.
    \label{Chain}
\end{align}
So we see that robust self-tests are the strongest while self-tests of states are a weaker form of self-testing. 
It is natural to ask for counterexamples showing that there are no equivalences above.

Unsurprisingly, it is often easier to prove that a certain game self-tests a reference strategy~$S$ but more care and effort is needed to show that this self-test is in fact robust. In most known cases, however, the same argument that is used to establish a self-testing result can also be turned robust by simply arguing that all the desired relations hold approximately rather than exactly. In fact, there are no known examples of self-tests that are not robust! Moreover, if a self-test is proven using a specific technique that leverages representation theory of finite groups then robustness of such self-tests follows from the Gowers-Hatami theorem \cite{GowersHatami, VidickGH}. The idea here is that all representations of finite groups are stable in a certain precise mathematical sense. Hence, one might hope that the same holds true for self-tests and any (exact) self-test is necessarily robust. Regretfully, we show that this is not the case, \emph{i.e.}, the converse of the first implication from~(\ref{Chain}) fails to hold:

\begin{theoremA}[Theorem \ref{thm:nonrobustselftest} \& Example \ref{ex:nonrobust}]\label{theoremA}
There exists a nonlocal game $\mathcal{G}$ that exactly self-tests a quantum strategy $S$, but this self-test is not robust.
\end{theoremA}

The second question we address in this article is whether every 2-party nonlocal game with quantum advantage self-tests some quantum state. The only known non-trivial set-ups that do not self-test any quantum state are known in the many party case \cite{GeoQCor}. We answer our second question in the negative:

\begin{theoremB}[Theorem \ref{thm:2pseud} \& Example \ref{ex:nostateselftest}]
    There exists a game $\mathcal{G}$ that does not self-test any state. 
\end{theoremB}

\paragraph{Discussion and outlook.} Another way to get different forms of self-testing is to vary the classical observations one has access to. In this article we have focused on self-testing where we assume that a quantum strategy achieves optimal or near-optimal winning probability in a nonlocal game. Another common form is where we derive self-testing from optimal or \mbox{near-optimal} Bell inequality violations. Since every nonlocal game can be cast as a Bell inequality, our counterexamples also hold in that setting. Finally, rather than self-testing from a single numerical value (\emph{e.g.}~Bell violation or winning probability), we can prove a self-testing result for quantum strategies that produce a certain quantum correlation (collection of probabilities). As we explain in more detail below, our counterexamples do not carry over to the setting of correlations. We summarize the known counterexamples to self-testing, including results of this article, in the table below:

\begin{table}[H]
\begin{center}
\begin{tabular}{lccc}
\toprule
&\makecell{Non-robust\\self-test}& \makecell{Not a measurement\\ self-test}&\makecell{Not a state\\ self-test}\\[.5ex]
\midrule
Nonlocal games&Theorem A&\cite{CMMNS}, \cite{slofstra2011lower}, \cite{Sigurd}&Theorem B\\
Bell inequalities&&\cite{kaniewski2020weak}, \cite{kaniewski2019maximal}&\cite{BASA}, \cite{FA}, \cite{GeoQCor}\\
Extreme correlations&?&\cite{Sigurd}, \cite{tavakoli2021mutually}&?\\
\bottomrule
\end{tabular}
\end{center}
\caption{Counterexamples for equivalence of different forms of self-testing in the three settings: nonlocal games, Bell inequalities, and extreme correlations. Note that every counterexample for nonlocal games also yields a counterexample for Bell inequalities.}\label{Table 1}
\end{table}

We now comment on Table \ref{Table 1}. We present the first example of a game that self-tests a quantum strategy, but this self-test is not robust (Theorem A) in this article. Since the correlations of the approximate optimal strategies of our game do not necessarily converge to the self-tested strategy, we do not get a non-robust self-test for correlations. We leave this as an open question.
For obtaining an example of a game that does not self-test measurements, one needs to construct games with inequivalent optimal strategies. In \cite{CMMNS} it was shown that the glued magic square game does not self-test measurements. It turns out that this game still self-tests a quantum state \cite{MNP}. This gives an example of a game that self-tests a state but not the whole strategy. Another, elementary construction of such a game is due to Storgaard \cite[Section 5.6]{Sigurd}: Take a game that is a self-test and add a question to Bob's question set for which the players always lose if Bob receives this question. The game is still a state self-test, but the operators Bob uses for the additional question are unknown to us. Another counterexample for measurement self-testing was obtained by Kaniewski \cite{kaniewski2020weak} in the framework of Bell inequalities. An extreme correlation obtained by inequivalent strategies was first presented in \cite{tavakoli2021mutually} and studied further in \cite{Sigurd}.
In this article, we construct the first examples of two-party games that do not self-test any state (Theorem B). As shown in the table above, an example of a three-party Bell inequaltiy that does not self-test any state has been identified in \cite{GeoQCor}. Furthermore, Bell inequalities that are maximally violated by states in an entangled subspace are studied in \cite{BASA} and \cite{FA}, but all their examples use more than two parties. We do not know an extreme correlation that does not self-test any state.

\paragraph{Proof ideas.} We now explain the key ideas behind our two main theorems. For obtaining the examples of non-robust self-tests and games that do not self-test any state, we introduce the \mbox{$(\mathcal{G}_1 \lor \mathcal{G}_2)$}-game for nonlocal games $\mathcal{G}_1$ and $\mathcal{G}_2$. In this game the players receive pairs of questions, one from each game. They choose which question they answer and win if they answered a question from the same game and additionally win this game with their answers. We believe that the \mbox{$(\mathcal{G}_1 \lor \mathcal{G}_2)$}-game may be of independent interest, beyond the applications for self-testing results.

For the non-robust self-test (Theorem A), we let $\mathcal{G}_1$ be a game that has no perfect quantum strategy, but a sequence of strategies whose winning probabilities converge to one. Note that Slofstra constructed such a game in \cite{SlofstraQcor}. We choose $\mathcal{G}_2$ to be a pseudo-telepathy game (\emph{i.e.} a game with perfect quantum strategy, but no perfect classical strategy) that self-tests some strategy $S_2$. Then the $(\mathcal{G}_1 \lor \mathcal{G}_2)$-game self-tests a strategy in which Alice and Bob always choose game $\mathcal{G}_2$ and play according to $S_2$. This self-test is not robust, however, as there are near-optimal strategies for the $(\mathcal{G}_1 \lor \mathcal{G}_2)$-game coming from the near-optimal strategies for $\mathcal{G}_1$. 



To obtain nonlocal games that do not self-test any state (Theorem B), it suffices to construct a game with two optimal quantum strategies that use states with coprime Schmidt ranks (see Lemma \ref{nonselftest}). In our case, those will be perfect quantum strategies. Note that we can get perfect strategies for the $(\mathcal{G}_1 \lor \mathcal{G}_2)$-game from perfect quantum strategies for $\mathcal{G}_1$ and $\mathcal{G}_2$, respectively, by Alice and Bob always choosing the same game and playing the perfect quantum strategy of this game. Thus, to get a game that does not self-test any state, it suffices to find two pseudo-telepathy games that have perfect strategies using states of coprime Schmidt ranks. Then the $(\mathcal{G}_1 \lor \mathcal{G}_2)$-game will not self-test any state. 

The magic square game is a pseudo-telepathy game with a strategy that uses a state with Schmidt rank $4$. Thus, we are left with finding a pseudo-telepathy game that has a perfect quantum strategy with a state of odd Schmidt rank. Our example of such a game is a $(G,t)$-independent set game, where $G$ is the orthogonality graph of a $3$-dimensional weak Kochen-Specker set and some $t \in \N$. The perfect quantum strategy uses a state of Schmidt rank $3$. 


\paragraph{Structure of the article.} In Section \ref{sect:prelim} we collect background material on nonlocal games. We then define all forms of self-tests in Section \ref{sect:selftest} and provide some known results we will use later on. Furthermore, we show that the synchronous magic square game self-tests a quantum strategy. In Section \ref{sect:newgame} we introduce the $(\mathcal{G}_1 \lor \mathcal{G}_2)$-game and establish some of its properties in case when both $\mathcal G_1$ and $\mathcal G_2$ are pseudo-telepathy games. Using the $(\mathcal{G}_1 \lor \mathcal{G}_2)$-game, we construct a non-robust self-test in Section \ref{sect:Nonrobustselftest}. Finally, we obtain games that do not self-test states in Section~\ref{sect:gamesnotselfteststates}.  

\section{Preliminaries}\label{sect:prelim}

We state some basic facts and notions we will use throughout this article. 
We write $[m]:=\{1, \dots, m\}$. Each Hilbert space $\mathcal{H}$ considered in this article is \textbf{finite-dimensional} which means that we have $\mathcal{H}\cong \C^d$ for some $d\in \N$. We let $\|\ket{\xi}\|=(\bra{\xi}\xi\rangle)^{\frac{1}{2}}$ for $\ket{\xi}\in \mathcal{H}$. An operator $X\in B(\mathcal{H})$ is \emph{positive}, if $\bra{\xi}X\ket{\xi}\geq 0$ for all $\ket{\xi}\in \mathcal{H}$. We write $X \leq Y$ if $Y-X$ is positive. 

A \emph{positive operator valued measure} (POVM) consists of a family of positive operators $\{M_i \in B(\mathcal{H})\,|\, i \in [m]\}$ such that $\sum_{i=1}^m M_i= 1_{B(\mathcal{H})}$. If all positive operators are projections ($M_i=M_i^*=M_i^2$), then we call $\{M_i \in B(\mathcal{H})\,|\, i \in [m]\}$ with $\sum_{i=1}^m M_i= 1_{B(\mathcal{H})}$ a \emph{projective measurement} (PVM).
A \emph{state} $\ket{\psi}$ is a unit vector in a Hilbert space $\mathcal{H}$. Each state $\ket{\psi}\in \mathcal{H}_A \otimes \mathcal{H}_B$ admits the so-called \emph{Schmidt decomposition} $\ket{\psi}=\sum_{i=1}^m\lambda_i a_i \otimes b_i$, where $\{a_i \,|\,i \in [m]\}$ and $\{b_i\,|\,i\in [m]\}$ are orthonormal sets in $\mathcal{H}_A$ and $\mathcal{H}_B$, respectively, and $\lambda_i \geq 0$ for all $i$. The strictly positive values $\lambda_i > 0$ are called \emph{Schmidt coefficients}. The \emph{Schmidt rank} of a state $\ket{\psi}$ is the number $n$ of Schmidt coefficients $\lambda_i$ (counted with multiplicity) in a Schmidt decomposition. We say that $\ket{\psi}$ has \emph{full Schmidt rank} if $n=\mathrm{dim}(\mathcal{H}_A)=\mathrm{dim}(\mathcal{H}_B)$. For $\ket{\psi}=\sum_{i=1}^m\lambda_i a_i \otimes b_i$, we let $\mathrm{supp}_A(\ket{\psi})=\mathrm{span}\{a_i\}_{i=1}^m$ and $\mathrm{supp}_B(\ket{\psi})=\mathrm{span}\{b_i\}_{i=1}^m$. Furthermore, we say that a subspace $\mathcal{K}\subseteq \mathcal{H}$ is an \emph{invariant subspace} of $X\in B(\mathcal{H})$ if $X(\mathcal{K})\subseteq \mathcal{K}$. 

In the following, we describe the framework of nonlocal games \cite{CHTW}. A two-player \emph{nonlocal game} $\mathcal{G}$ is played between two collaborating players, Alice and Bob, and a referee. It is specified by finite input sets $I_A$, $I_B$ and finite output sets $O_A$, $O_B$ for Alice and Bob, respectively as well as a \emph{verification function} $V:I_A\times I_B \times O_A \times O_B\to \{0,1\}$ and a probability distribution $\pi$ on $I_A \times I_B$. In the game, the referee samples a pair $(x,y) \in I_A\times I_B$ using the distribution $\pi$ and sends $x$ to Alice and $y$ to Bob. Alice and Bob respond with $a\in O_A$ and $b\in O_B$, respectively. They win if $V(x, y,a,b)=1$. In our examples, we always use the uniform distribution $\pi$. 
The players are not allowed to communicate during the game, but they can agree on a strategy beforehand. A \emph{classical strategy} for a nonlocal game is a strategy in which Alice and Bob only have access to shared randomness. In a \emph{quantum strategy}, the players are allowed to perform local measurements on a shared entangled state. Thus, a quantum strategy may be written as $S=(\ket{\psi} \in \mathcal{H}_A\otimes \mathcal{H}_B, \{E_{xa}\}_x, \{F_{yb}\}_y)$, where $\ket{\psi}$ is the shared entangled state, $\{E_{xa}\}_x$ are the POVMs for each $x \in I_A$ for Alice and $\{F_{yb}\}_y$ are the POVMs for each $y \in I_B$ for Bob. 
A nonlocal game $\mathcal{G}$ is called \emph{synchronous} \cite{paulsen2016estimating} if we have $I_A=I_B$, $O_A=O_B$ and if Alice and Bob receive identical inputs, they must answer with identical outputs to win, i.e. $V(x,x,a,b)=0$ for $a\neq b$. 

The \emph{classical value} $\omega(\mathcal{G})$ of a nonlocal game $\mathcal{G}$ is the greatest probability of winning the game with a classical strategy. The \emph{quantum value} $\omega^*(\mathcal{G})$ is the supremum over the winning probabilities of the quantum strategies for the game. An \emph{optimal quantum strategy} is a quantum strategy achieving this quantum value. In contrary to the classical value, the quantum value is not always attained \cite{SlofstraQcor}. In general, the quantum value is bigger than the classical value. We are especially interested in \emph{pseudo-telepathy} games, which are games having a \emph{perfect} quantum strategy (a strategy that allows the players to win with probability one), but no perfect classical strategy.
 
\section{Self-testing}\label{sect:selftest}

In this section, we formally introduce the notion of self-testing. The idea of self-testing is the following. A game self-tests an optimal quantum strategy $\tilde{S}$ if the strategy is unique, in the sense that every other optimal strategy $S$ is related the strategy $\tilde{S}$ by a local isometry. As in \cite{MPS}, we will say that $\tilde{S}$ is a local dilation of $S$.

\begin{definition}
Let $\tilde{S}=(\tilde{\ket{\psi}} \in \tilde{\mathcal{H}}_A\otimes \tilde{\mathcal{H}}_B, \{\tilde{E}_{xa}\}_x, \{\tilde{F}_{yb}\}_y)$, $S=(\ket{\psi} \in \mathcal{H}_A\otimes \mathcal{H}_B, \{E_{xa}\}_x,\allowbreak \{F_{yb}\}_y)$ be quantum strategies. We say that $\tilde{S}$ is a \emph{local dilation} of $S$ if there exist Hilbert spaces $\mathcal{H}_{A,aux}$ and $\mathcal{H}_{B,aux}$, a state $\ket{aux}\in \mathcal{H}_{A,aux}\otimes\mathcal{H}_{B,aux}$ and isometries $U_A: \mathcal{H}_A \to \tilde{\mathcal{H}}_A \otimes \mathcal{H}_{A,aux}$, $U_B:\mathcal{H}_B \to \tilde{\mathcal{H}}_B \otimes \mathcal{H}_{B,aux}$ such that with $U:=U_A \otimes U_B$ it holds
\begin{align}
    U\ket{\psi}&= \ket{\tilde{\psi}}\otimes \ket{aux},\label{selfteststate}\\
    U(E_{xa}\otimes 1)\ket{\psi}&=[ (\tilde{E}_{xa}\otimes 1)\ket{\tilde{\psi}}]\otimes \ket{aux},\label{selftestmA}\\
    U(1 \otimes F_{yb})\ket{\psi}&=[(1 \otimes \tilde{F}_{yb})\ket{\tilde{\psi}}]\otimes \ket{aux}\label{selftestmB}.
\end{align}
\end{definition}

Note that if $\tilde{S}$ is a local dilation of $S$, then they induce the same correlation and thus the same winning probability for a game $\mathcal{G}$. 

\begin{remark}
There is a slight abuse of notation in the previous Definition. The lefthand side of equations \eqref{selfteststate}, \eqref{selftestmA} and \eqref{selftestmB} is an element of $\tilde{\mathcal{H}}_A \otimes \mathcal{H}_{A,aux} \otimes \tilde{\mathcal{H}}_B \otimes \mathcal{H}_{B,aux}$, whereas the righthand side is an element of  $\tilde{\mathcal{H}}_A\otimes \tilde{\mathcal{H}}_B  \otimes\mathcal{H}_{A,aux}  \otimes \mathcal{H}_{B,aux}$. We identify those Hilbert spaces via the unitary that flips the second and the third tensor factors.
\end{remark}

\begin{remark}
The equations \eqref{selftestmA} and \eqref{selftestmB} are equivalent to     
\begin{align*}
    U(E_{xa}\otimes F_{yb})\ket{\psi}&=[ (\tilde{E}_{xa}\otimes \tilde{F}_{yb})\ket{\tilde{\psi}}]\otimes \ket{aux}.
\end{align*}
\end{remark}

\begin{definition}
Let $\mathcal{G}$ be a nonlocal game and $\tilde{S}=(\tilde{\ket{\psi}} \in \tilde{\mathcal{H}}_A\otimes \tilde{\mathcal{H}}_B, \{\tilde{E}_{xa}\}_x, \{\tilde{F}_{yb}\}_y)$ be an optimal quantum strategy. We say that $\mathcal{G}$ \emph{self-tests} the strategy $\tilde{S}$ if $\tilde{S}$ is a local dilation of any optimal quantum strategy $S=(\ket{\psi} \in \mathcal{H}_A\otimes \mathcal{H}_B, \{E_{xa}\}_x, \{F_{yb}\}_y)$.
If only equation \eqref{selfteststate} holds for all optimal quantum strategies $S$, we say that $\mathcal{G}$ \emph{self-tests the state} $\ket{\tilde{\psi}}$.
\end{definition}

\begin{remark}
Some authors only impose that $\tilde{S}$ is a local dilation of any projective strategy $S$. This is a priori a weaker notion of self-testing. We will see, however, that in some cases (e.g. the magic square game) the PVM self-test can be used to prove a POVM self-test for the synchronous version of the game.    
\end{remark}

We will now state some results of \cite{MPS} on local dilations. Those will mostly help to reduce proving self-testing for general quantum strategies to quantum strategies having states with full Schmidt rank. 

\begin{lemma}\cite[Lemma 4.8]{MPS}\label{lem:invariance}
Let $X\in B(\mathcal{H}_A)$, $Y\in B(\mathcal{H}_B)$ and $\ket{\psi}\in \mathcal{H}_A\otimes \mathcal{H}_B$ such that $(X\otimes 1)\ket{\psi}=(1 \otimes Y)\ket{\psi}$. Then $\mathrm{supp}_A(\ket{\psi})$ is invariant under $X$ and $\mathrm{supp}_B(\ket{\psi})$ is invariant under $Y$.
\end{lemma}

The next lemma and its corollary split \cite[Lemma 4.9]{MPS} into two parts, making the lemma slightly more general. We do this since we need to use the lemma in its more general form later on. 

\begin{lemma}\label{lem:DilationfullSchmidtrank}
Let $\mathcal{G}$ be a nonlocal game. Let $S=(\ket{\psi}\in \C^{d_A}\otimes \C^{d_B},\{E_{xa}\}_{x},\{F_{yb}\}_{y})$ be a quantum strategy such that $\mathrm{supp}_A(\ket{\psi})$ and $\mathrm{supp}_B(\ket{\psi})$ are invariant under each $E_{xa}$ and $F_{yb}$, respectively. Then there exists a quantum strategy $S'=(\ket{\psi'},\{E_{xa}'\}_{x},\{F_{yb}'\}_{y})$ such that $\ket{\psi'}$ has full Schmidt rank and $S'$ is a local dilation of $S$. 
\end{lemma}

\begin{proof}
We follow the proof of \cite[Lemma 4.9]{MPS}. Consider the Schmidt decomposition $\ket{\psi}=\sum_{i=1}^r \alpha_i \xi_i \otimes \eta_i$ and let $U_A:\C^r\to \C^{d_A}$, $U_B:\C^r\to \C^{d_B}$ be isometries given by
\begin{align*}
    U_A=\sum_{i=1}^r\xi_i e_i^*, \quad U_B=\sum_{i=1}^r\eta_i e_i^*.
\end{align*}
Then
\begin{align*}
    E_{xa}'=U_A^*E_{xa}U_A, \quad F_{yb}'=U_B^*F_{yb}U_B
\end{align*}
are positive operators and $\{E_{xa}'\,|\,x\in [r]\}$, $\{F_{yb}'\,|\,y\in [r]\}$ are POVMs. With $\ket{\psi'}=(U_A^*\otimes U_B^*)\ket{\psi}$, we get a quantum strategy $S'$, where $\ket{\psi'}$ has full Schmidt rank. 

We will now show that $S'$ is a local dilation of $S$. Set $H_{A,aux}=\C^{d_A}$, $H_{B,aux}=\C^{d_B}$ and define isometries $V_A:\C^{d_A}\to \C^r \otimes \C^{d_A}$, $V_B:\C^{d_B}\to \C^r \otimes \C^{d_B}$ by
\begin{align*}
    V_A(v)&=U_A^*(v)\otimes \xi_1 + e_1\otimes (1_{d_A}-U_AU_A^*)(v),\\
    V_B(w)&=U_B^*(w)\otimes \eta_1 + e_1\otimes (1_{d_B}-U_BU_B^*)(w),
\end{align*}
where $v\in \C^{d_A}$, $w\in \C^{d_B}$. Since $\mathrm{supp}_A(\ket{\psi})$ and $\mathrm{supp}_B(\ket{\psi})$ are invariant under each $E_{xa}$ and $F_{yb}$, respectively, and $U_AU_A^*$ and $U_BU_B^*$ are the projections onto $\mathrm{supp}_A(\ket{\psi})$ and $\mathrm{supp}_B(\ket{\psi})$, respectively, we have
\begin{align*}
    (V_A\otimes V_B)(E_{xa}\otimes F_{yb})\ket{\psi}=(U_A^*\otimes U_B^*)((E_{xa}\otimes F_{yb})\ket{\psi})\otimes (\xi_1 \otimes \eta_1).
\end{align*}
Furthermore, it holds
\begin{align*}
    (U_A^*\otimes U_B^*)(E_{xa}\otimes F_{yb})\ket{\psi}
    &=(U_A^*\otimes U_B^*)(E_{xa}\otimes F_{yb})(U_AU_A^*\otimes U_BU_B^*)\ket{\psi}\\
    &=(E_{xa}'\otimes F_{yb}')\ket{\psi'},
\end{align*}
since $U_AU_A^*$ and $U_BU_B^*$ are the projections onto $\mathrm{supp}_A(\ket{\psi})$ and $\mathrm{supp}_B(\ket{\psi})$, respectively. This concludes the proof. 
\end{proof}

\begin{cor}\label{fullSchmidtsync}
Let $\mathcal{G}$ be a synchronous game and let $S=(\ket{\psi},\{E_{xa}\}_{x},\{F_{yb}\}_{y})$ a perfect quantum strategy of $\mathcal{G}$. Then there exists a perfect quantum strategy $S'=(\ket{\psi'},\{E_{xa}'\}_{x},\allowbreak\{F_{yb}'\}_{y})$ of $\mathcal{G}$ such that $\ket{\psi'}$ has full Schmidt rank and $S'$ is a local dilation of $S$. 
\end{cor}

\begin{proof}
By \cite[Corollary 3.6 (a)]{MPS}, we know $(E_{xa}\otimes 1)\ket{\psi}=(1\otimes F_{xa})\ket{\psi}$. Lemma \ref{lem:invariance} yields that $\mathrm{supp}_A(\ket{\psi})$ and $\mathrm{supp}_B(\ket{\psi})$ are invariant under each $E_{xa}$ and $F_{yb}$, respectively. We obtain the result by Lemma \ref{lem:DilationfullSchmidtrank}. 
\end{proof}

We also have a transitivity result for local dilations. One can use it for proving self-testing results in the follwing way. First show that a general optimal quantum strategy dilates to a quantum strategy having a state of full Schmidt rank. Then check if the reference strategy is a local dilation of any optimal strategy with a state of full Schmidt rank. 

\begin{lemma}\cite[Lemma 4.7]{MPS}\label{transitivity}
Let $S_1$, $S_2$ and $S_3$ be quantum strategies. If $S_1$ is a local dilation of $S_2$ and $S_2$ is a local dilation of $S_3$, then $S_1$ is a local dilation of $S_3$.
\end{lemma}

Finally, we give the definition of robust self-testing. Roughly speaking, a game robustly self-tests a quantum strategy $\tilde{S}$ if it self-tests the strategy and additionally, for every almost optimal strategy $S$, a local dilation of $S$ is close to $\tilde{S}$. For $\delta>0$, we say that a quantum strategy is \emph{$\delta$-optimal} for a game $\mathcal{G}$ if the winning probability is greater or equal to $\omega^*(\mathcal{G})-\delta$, where $\omega^*(\mathcal{G})$ is the quantum value of $\mathcal{G}$.

\begin{definition}
Let $\mathcal{G}$ be a nonlocal game, $\tilde{S}=(\tilde{\ket{\psi}} \in \tilde{\mathcal{H}}_A\otimes \tilde{\mathcal{H}}_B, \{\tilde{E}_{xa}\}_x, \{\tilde{F}_{yb}\}_y)$ an optimal quantum strategy. Then $\mathcal{G}$ is a \emph{robust self-test} for $\tilde{S}$ if $\mathcal{G}$ is a self-test for $\tilde{S}$ and for any $\varepsilon>0$, there exists a $\delta>0$ such that for any $\delta$-optimal strategy $S=(\ket{\psi} \in \mathcal{H}_A\otimes \mathcal{H}_B, \{E_{xa}\}_x, \{B_{yb}\}_y)$, there exists Hilbert spaces $\mathcal{H}_{A,aux}$ and $\mathcal{H}_{B,aux}$, a state $\ket{aux}\in \mathcal{H}_{A,aux}\otimes\mathcal{H}_{B,aux}$ and isometries $U_A: \mathcal{H}_A \to \tilde{\mathcal{H}}_A \otimes \mathcal{H}_{A,aux}$, $U_B:\mathcal{H}_B \to \tilde{\mathcal{H}}_B \otimes \mathcal{H}_{B,aux}$ such that with $U:=U_A \otimes U_B$ it holds
\begin{align*}
\|U\ket{\psi}- \tilde{\ket{\psi}}\otimes \ket{aux}\|&\leq \varepsilon,\\
\|U(E_{xa}\otimes 1)\ket{\psi}- [(\tilde{E}_{xa}\otimes 1)\tilde{\ket{\psi}}]\otimes \ket{aux}\|&\leq \varepsilon,\nonumber\\
\|U(1 \otimes \tilde{F}_{yb})\ket{\psi}- [(1\otimes \tilde{F}_{yb})\tilde{\ket{\psi}}]\otimes \ket{aux}\|&\leq \varepsilon.
\end{align*}
\end{definition}

\subsection{The synchronous magic square game self-tests a quantum strategy}\label{syncmagicsquare}

In this subsection, we show that the synchronous magic square game self-tests a quantum strategy $\tilde{S}$ with maximally entangled state $\ket{\psi_4}$. In \cite{WBMS} it was shown that the magic square game self-tests a strategy $S$. Note that their general strategies consist of projective measurements, the case of strategies with POVMs is not considered. We will deduce a POVM self-testing result for the synchronous magic square game from the PVM self-test of the magic square. 

We will first describe the magic square game of \cite[Section 5]{BBT} and then its synchronous version. Note that there are also other versions of the magic square game. We choose the one with minimal input and output sets. 
Consider the following set of equations
\begin{align*}
    &x_1x_2x_3=1 \quad(r_1), &&x_1x_4x_7=-1\quad(c_1),\\
    &x_4x_5x_6=1\quad(r_2), &&x_2x_5x_8=-1\quad(c_2),\\
    &x_7x_8x_9=1\quad(r_3), &&x_3x_6x_9=-1\quad(c_3).
\end{align*}
In the magic square game, the referee sends one of the three equations on the left hand side to Alice and one of the three equations on the right hand side to Bob. Alice answers with a $\{-1,1\}$-assignment of the variables such that their product is $1$ and Bob answers with a $\{-1,1\}$-assignment of the variables such that their product is $-1$. The players win the game, if their assignments coincide in the common variable of the equations they received. More formally, we have
\begin{align*}
    &I_A=R:=\{r_1,r_2,r_3\},\\
    &I_B=C:=\{c_1,c_2,c_3\},\\
    &O_A=O_1:=\{(1,1,1),(1,-1,-1), (-1,1,-1), (-1,-1,1)\},\\
    &O_B=O_{-1}:=\{(-1,-1,-1),(-1,1,1), (1,-1,1), (1,1,-1)\},\\
    &V_{MS}(r_i, c_j, a, b)=\begin{cases}1 \text{ if }a_j=b_i, \\0 \text{ otherwise.}\end{cases}
\end{align*}
Here $a=(a_1,a_2,a_3), b=(b_1,b_2,b_3)$. Note that it is shown in \cite[Section 5]{BBT} that this game is a pseudo-telepathy game.

The synchronous version of the magic square game is played as follows. The referee sends one of the six equations above to Alice and one of them to Bob. They answer with $\{-1,1\}$-assignments such that the equations are fulfilled. Alice and Bob win the game if their assignments coincide on the common variables of the equations they received. Here, we have
\begin{align*}
    &I_A=I_B=R\cup C, \quad O_A=O_B=O_1\cup O_{-1},\\
    &V(x, y, a, b)=\begin{cases}
    V_{MS}(x, y, a, b)\text{ if }x\in R, y\in C, a\in O_1, b\in O_{-1},\\
    V_{MS}(y, x, b, a)\text{ if }x\in C, y\in R, a\in O_{-1}, b\in O_1,\\
    1 \text{ if }(x,y\in R, x\neq y, a,b \in O_1) \text{ or }(x,y\in C, x\neq y, a,b \in O_{-1}),\\
    \delta_{ab}\text{ if }x=y\in R, a,b \in O_1 \text{ or } x=y\in C, a,b \in O_{-1},\\
    0 \text{ otherwise.}\end{cases}
\end{align*}

The following self-testing statement for the magic square game was shown in \cite{WBMS}. We will use this theorem to deduce a self-testing statement for the synchronous magic square game. Recall the Pauli matrices
\begin{align*}
\sigma_X=\begin{pmatrix}0&1\\1&0\end{pmatrix},\,  \sigma_Y=\begin{pmatrix}0&-i\\i&0\end{pmatrix} \text{ and } \sigma_Z=\begin{pmatrix}1&0\\0&-1\end{pmatrix}.
\end{align*}

\begin{theorem}[\cite{WBMS}]\label{thm:projselftestms}
Restricting to projective measurements, the magic square game self-tests a perfect quantum strategy $S=(\ket{\psi_4}, \{E_{xa}\}_x,\{F_{yb}\}_y)$. Here $\ket{\psi_4}=\frac{1}{2}\sum_{i=1}^4 e_i \otimes e_i$ and $E_{xa}=\frac{1}{8}(1+a_1X_1)(1+a_2X_2)(1+a_3X_3)$, $F_{yb}=\frac{1}{8}(1+b_1Y_1^{\tp})(1+b_2Y_2^{\tp})(1+b_3Y_3^{\tp})$, where $X_i$ and $Y_i$ are the i-th entries in row $x$ and column $y$ of
\begin{center}\begin{tikzpicture}[scale=0.7]
\draw (-4.5,1.5)--(4.5,1.5);\draw (-4.5,0.5)--(4.5,0.5);\draw (-4.5,-0.5)--(4.5,-0.5);\draw (-4.5,-1.5)--(4.5,-1.5);
\draw (-4.5,1.5)--(-4.5,-1.5);\draw (4.5,1.5)--(4.5,-1.5);\draw (1.5,1.5)--(1.5,-1.5);\draw (-1.5,1.5)--(-1.5,-1.5);
 \draw (-3,1) node{$I\otimes \sigma_Z$}; \draw (0,1) node{$\sigma_Z\otimes I$};\draw (3,1) node{$\sigma_Z\otimes \sigma_Z$};
 \draw (-3,0) node{$\sigma_X\otimes I$}; \draw (0,0) node{$I\otimes \sigma_X$};\draw (3,0) node{$\sigma_X\otimes \sigma_X$};
 \draw (-3,-1) node{$-\sigma_X\otimes \sigma_Z$}; \draw (0,-1) node{$-\sigma_Z\otimes \sigma_X$};\draw (3,-1) node{$\sigma_Y\otimes \sigma_Y$};
\end{tikzpicture}\end{center}
\end{theorem}

\begin{cor}\label{cor:Selftestsyncmagicsquare}
Let $S$ be as in Theorem \ref{thm:projselftestms}. The synchronous magic square game self-tests the perfect quantum strategy $\tilde{S}=(\ket{\psi_4}, \{\tilde{E}_{xa}\}_x,\{\tilde{F}_{yb}\}_y)$, where 
\begin{align*}
    &\tilde{E}_{xa}=\begin{cases}E_{xa} \text{ if } x \in R, a \in O_1\\ F_{xa}^{\tp}\text{ if } x \in C, a \in O_{-1},\\ 0 \text{ otherwise,}\end{cases}
    &&\tilde{F}_{yb}=\begin{cases}F_{yb}\text{ if } y \in C, b \in O_{-1}\\ E_{yb}^{\tp}\text{ if } y \in R, b \in O_1,\\ 0 \text{ otherwise.}\end{cases}
\end{align*}
\end{cor}

\begin{proof}
Let $S'=(\ket{\psi'}, \{E_{xa}'\},\{F_{yb}'\})$ be a perfect quantum strategy for the synchronous magic square game. By Corollary \ref{fullSchmidtsync}, there exists a perfect quantum strategy $\hat{S}=(\hat{\ket{\psi}}, \{\hat{E}_{xa}\},\{\hat{F}_{yb}\})$ that is a local dilation of $S'$, where $\hat{\ket{\psi}}$ has full Schmidt rank, $\hat{\ket{\psi}}=\sum_{i=1}^d\hat{\lambda}_i e_i \otimes e_i$ with $\hat{\lambda}_i >0$ for all $i \in [d]$. Furthermore, by \cite[Corollary 3.6]{MPS}, the operators $\hat{E}_{xa}$ and $\hat{F}_{yb}$ are projections. 

Let $\hat{\varphi}:=\mathrm{diag}(\hat{\lambda}_i) \in \C^{d\times d}$, i.e. $\hat{\varphi}$ is a diagonal matrix with entries $\hat{\lambda}_i$. Let $x \in R$ and $a \in O_{-1}$. Then 
\begin{align*}
    \mathrm{Tr}(\hat{E}_{xa}\hat{\varphi}^2)=\bra{\hat{\psi}}\hat{E}_{xa}\otimes 1 \ket{\hat{\psi}}=\bra{\hat{\psi}}\hat{E}_{xa}\otimes \sum_{b}\hat{F}_{yb} \ket{\hat{\psi}}=0,
\end{align*}
since $\bra{\hat{\psi}}\hat{E}_{xa}\otimes \hat{F}_{yb} \ket{\hat{\psi}}=0$ for all $y,b$ as $a \in O_{-1}$. Since $\hat{\varphi}$ is invertible, we conclude $\hat{E}_{xa}=0$. We similarly get $\hat{E}_{xa}=0$ for $x\in C, a\in O_1$ and $\hat{F}_{yb}=0$ for $y\in C, b \in O_1$ or $y\in R$, $b\in O_{-1}$.

By the previous argument, we get that $\{\hat{E}_{xa}\,|\, a\in O_1\}$ and  $\{\hat{F}_{yb}\,|\, b\in O_{-1}\}$ are PVMs for $x\in R$, $y\in C$. Together with $\hat{\ket{\psi}}$, they form a perfect quantum strategy for the magic square game. By Theorem \ref{thm:projselftestms}, there exists
Hilbert spaces $\mathcal{H}_{A,aux}$ and $\mathcal{H}_{B,aux}$, a state $\ket{aux}\in \mathcal{H}_{A,aux}\otimes\mathcal{H}_{B,aux}$ and isometries $U_A: \hat{\mathcal{H}}_A \to \mathcal{H}_A \otimes \mathcal{H}_{A,aux}$, $U_B:\hat{\mathcal{H}}_B \to \mathcal{H}_B \otimes \mathcal{H}_{B,aux}$ such that with $U:=U_A \otimes U_B$ it holds
\begin{align}
    U\ket{\hat{\psi}}&= \ket{\psi_4}\otimes \ket{aux},\nonumber\\
    U(\hat{E}_{xa}\otimes 1)\ket{\hat{\psi}}&=[(E_{xa}\otimes 1)\ket{\psi_4}]\otimes \ket{aux},\nonumber\\
    U(1\otimes \hat{F}_{yb} )\ket{\hat{\psi}}&=[(1\otimes F_{xa})\ket{\psi_4}]\otimes \ket{aux}\label{partselftest}
\end{align}
for $x\in R$, $y\in C$ and $a \in O_1$, $b \in O_{-1}$. 

Now, let $x \in C$ and $a \in O_{-1}$. By \cite[Corollary 3.6 (a)]{MPS}, we have
\begin{align*}
     U(\hat{E}_{xa}\otimes 1)\ket{\hat{\psi}}=U(1 \otimes\hat{F}_{xa})\ket{\hat{\psi}}.
\end{align*}
Furthermore, it holds
\begin{align*}
    U(1 \otimes\hat{F}_{xa})\ket{\hat{\psi}}=[(1\otimes F_{yb})\ket{\psi_4}]\otimes \ket{aux}
\end{align*}
by \eqref{partselftest}. We conclude 
\begin{align}
     U(\hat{E}_{xa}\otimes 1)\ket{\hat{\psi}}=[(1\otimes F_{yb})\ket{\psi_4}]\otimes \ket{aux}=[(F_{xa}^{\tp}\otimes 1)\ket{\psi_4}]\otimes \ket{aux}\label{part2selftest}
\end{align}
for all $x\in C$ and $a \in O_{-1}$. One similarly obtains 
\begin{align}
     U(1\otimes \hat{F}_{yb} )\ket{\hat{\psi}}= [(1\otimes E_{yb}^{\tp})\ket{\psi_4}]\otimes \ket{aux}\label{part3selftest}
\end{align}
for all $y\in R$ and $a \in O_1$. Since we know $\hat{E}_{xa}=0=E_{xa}$ and $\hat{F}_{yb}=0=F_{yb}$ for the remaining $x,y,a,b$, we deduce from \eqref{partselftest}, \eqref{part2selftest} and \eqref{part3selftest} that the synchronous magic square game self-tests the perfect quantum strategy $\tilde{S}$.
\end{proof}

\section[x]{The $(\mathcal{G}_1\lor\mathcal{G}_2)$-game}\label{sect:newgame}

Let $\mathcal{G}_1$ and $\mathcal{G}_2$ be nonlocal games. The $(\mathcal{G}_1\lor \mathcal{G}_2)$-game is played as follows: The referee sends Alice and Bob a pair of questions $(x_1, x_2)$ and $(y_1, y_2)$, respectively, where $x_1, y_1$ are questions in $\mathcal{G}_1$ and $x_2, y_2$ in $\mathcal{G}_2$. Each of them chooses one of the questions they received and responds with an answer from the corresponding game. To win the game, two conditions have to be fulfilled:
\begin{itemize}
    \item[(1)] Alice and Bob have to choose questions from the same game,
    \item[(2)] their answers have to win the corresponding game. 
\end{itemize}
More formally, suppose the nonlocal games $\mathcal{G}_1$ and $\mathcal{G}_2$ have input sets $I_{A,i}$, $I_{B,i}$, output sets $O_{A,i}$, $O_{B,i}$, verification functions $V_{\mathcal{G}_i}$ and probability distributions $\pi_i$ on $I_{A,i}\times I_{B,i}$ for $i=1,2$. Then the $(\mathcal{G}_1\lor \mathcal{G}_2)$-game has input sets $I_{A,1}\times I_{A,2}$, $I_{B,1}\times I_{B,2}$, output sets $O_{A,1} \dot\cup O_{A,2}$, $O_{B,1}\dot\cup O_{B,2}$ and verification function
\begin{align*}
    V((x_1, x_2), (y_1,y_2), a,b)=\begin{cases}V_{\mathcal{G}_i}(x_i,y_i,a,b) \,\text{ if } a \in O_{A,i} \text{ and } b \in O_{B,i} \text{ for some $i=1,2$,} \\0 \, \text{ otherwise.}\end{cases}
\end{align*}
For the probability distribution, we take $\pi=\pi_1 \times \pi_2$ on $(I_{A,1}\times I_{A,2})\times (I_{B,1}\times I_{B,2})$, \emph{i.e.} $\pi((x_1,x_2),(y_1,y_2))=\pi_1(x_1,y_1)\pi_2(x_2,y_2)$. In this article, we assume that the probability distributions of $\mathcal{G}_1$ and $\mathcal{G}_2$ are uniform, so this will also be the case for the $(\mathcal{G}_1\lor \mathcal{G}_2)$-game. The $(\mathcal{G}_1\lor \mathcal{G}_2)$-game will be used for our counterexamples in this article, we believe that it can also be of independent interest. 

We first check that the $(\mathcal{G}_1\lor\mathcal{G}_2)$-game keeps its quantum advantage if one of the games is a pseudo-telepathy game and the other one does not have a perfect quantum strategy. 

\begin{lemma}\label{G1orG2pseudotele}
Let $\mathcal{G}_1$ be a pseudo-telepathy game and let $\mathcal{G}_2$ be a game with $\omega(G)<1$. Then the $(\mathcal{G}_1\lor \mathcal{G}_2)$-game is a pseudo-telepathy game.
\end{lemma}

\begin{proof}
It is easy to see that the $(\mathcal{G}_1\lor \mathcal{G}_2)$-game has a perfect strategy. Indeed, Alice and Bob can always choose to answer  the question from $\mathcal{G}_1$ and then use a perfect quantum strategy for $\mathcal{G}_1$ to come up with their answers.

It remains to show that there is no perfect classical strategy for the $(\mathcal{G}_1\lor \mathcal{G}_2)$-game. First note that it suffices to show that there is no perfect deterministic strategy. For contradiction assume that there is a perfect deterministic strategy $S$. If in this strategy Alice chooses a question $x_1$ for some pair $(x_1, x_2)$, then we know that Bob has to choose $y_1$ for all pairs $(y_1,y_2)$, as otherwise there are questions for which the players lose the game. A similar argument shows that Alice also always has to choose the questions from $\mathcal{G}_1$. 
This shows that in strategy $S$ the players always choose to answer questions from $\mathcal{G}_1$. Hence, $S$ can be used to construct a perfect classical strategy for game $\mathcal{G}_1$ which is a contradiction. 
Since $\mathcal{G}_2$ also does not admit a perfect classical strategy, a similar argument works if in strategy $S$ Alice chooses a question $x_2$ for some tuple $(x_1, x_2)$.
\end{proof}

The following lemma will be used to show that restricting a POVM measurement with a projection yields a POVM measurement on a smaller Hilbert space.  

\begin{lemma}\cite[Proposition II, 3.3.2]{Bla} \label{projeq1}
Let $A$ be a $C^*$-algebra, let $a\in A$ be a positive element and $p\in A$ a projection. If $a\leq p$, then $ap=pa=a$.
\end{lemma}

The next lemma shows that if we have a perfect quantum strategy for the $(\mathcal{G}_1\lor\mathcal{G}_2)$-game, where Alice and Bob have non-zero probability of answering with some outputs of $\mathcal{G}_1$, then there is a perfect quantum strategy for $\mathcal{G}_1$. In particular, this shows that if $\mathcal{G}_2$ has a perfect quantum strategy and $\mathcal{G}_1$ does not, the players will always choose to play $\mathcal{G}_2$ and never $\mathcal{G}_1$.

\begin{lemma}\label{perfectqstratG1}
Let $\mathcal{G}_1$ and $\mathcal{G}_2$ be nonlocal games and consider the $(\mathcal{G}_1\lor \mathcal{G}_2)$-game as above. If there is a perfect quantum strategy $S=(\ket{\psi}, \{E_{(x_1,x_2)a}\}_{(x_1,x_2)},\{F_{(y_1,y_2)b}\}_{(y_1,y_2)})$ of the $(\mathcal{G}_1\lor \mathcal{G}_2)$-game, where additionally 
\begin{align}\label{lemmaassumption}
    \bra{\psi}E_{(x_1,x_2)a}\otimes F_{(y_1,y_2)b}\ket{\psi}> 0
\end{align}
for some $a \in O_{A,1}$ and $b \in O_{B,1}$, then $\mathcal{G}_1$ has a perfect quantum strategy. 
\end{lemma}

\begin{proof}
Let $S'=(\ket{\psi'} \in \mathcal{H}_A'\otimes \mathcal{H}_B', \{E_{(x_1,x_2)a}'\}_{(x_1,x_2)},\{F_{(y_1,y_2)b}'\}_{(y_1,y_2)})$ be a perfect quantum strategy of the $(\mathcal{G}_1\lor \mathcal{G}_2)$-game. By restricting the state and the operators from $S'$ to $\mathrm{supp}_A(\ket{\psi'})\otimes \mathrm{supp}_B(\ket{\psi'})$, we get a perfect quantum strategy $S=(\ket{\psi}\in \mathcal{H}_A\otimes \mathcal{H}_B, \{E_{(x_1,x_2)a}\}_{(x_1,x_2)},$\linebreak$\{F_{(y_1,y_2)b}\}_{(y_1,y_2)})$ such that
\begin{align}
    \bra{\psi}E_{(x_1,x_2)a}\otimes F_{(y_1,y_2)b}\ket{\psi}=\bra{\psi'}E_{(x_1,x_2)a}'\otimes F_{(y_1,y_2)b}'\ket{\psi'}\label{same}
\end{align}
and $\ket{\psi}$ has full Schmidt rank, \emph{i.e.} $\ket{\psi}=\sum_{i=1}^d\lambda_i e_i \otimes e_i$ with $\lambda_i >0$ for all $i \in [d]$. We will first show 
\begin{align*}
    \sum_{a \in O_{A,i}}E_{(x_1,x_2)a}=\sum_{a \in O_{A,i}}E_{(s_1,s_2)a}, \qquad
    \sum_{b \in O_{B,j}}F_{(y_1,y_2)b}=\sum_{b \in O_{B,j}}F_{(t_1,t_2)b}
\end{align*}
for all $(x_1,x_2), (s_1,s_2) \in I_{A,1}\times I_{A,2}$, $i=1,2$ and $(y_1,y_2), (t_1,t_2) \in I_{B,1}\times I_{B,2}$, $j=1,2$. Note that we have 
\begin{align}
    \sum_{a \in O_{A,1}\dot\cup O_{A,2}}E_{(x_1,x_2)a}=1= \sum_{b \in O_{B,1}\dot\cup O_{B,2}}F_{(y_1,y_2)b}\label{equal1}
\end{align}
and $\bra{\psi}E_{(x_1,x_2)a} \otimes F_{(y_1,y_2)b}\ket{\psi}=0$ for all $a\in O_{A,i}$, $b\in O_{B,j}$ with $i\neq j$, since $S$ is a perfect quantum strategy.
Let $\varphi:=\mathrm{diag}(\lambda_i) \in \C^{d\times d}$, \emph{i.e.} $\varphi$ is the diagonal matrix with entries $\lambda_i$. Let $p_{(x_1,x_2)i}:=\sum_{a \in O_{A,i}}E_{(x_1,x_2)a}$, $i=1,2$ and $q_{(y_1,y_2)j}:=\sum_{b \in O_{B,j}}F_{(y_1,y_2)b}$, $j=1,2$. It holds
\begin{align*}
\mathrm{Tr}(p_{(x_1,x_2)i} \varphi (q_{(y_1,y_2)j})^{\tp} \varphi)
&=\mathrm{Tr}(\varphi^*p_{(x_1,x_2)i} \varphi (q_{(y_1,y_2)j})^{\tp})\\
&=\bra{\psi}\sum_{a \in O_{A,i}}E_{(x_1,x_2)a} \otimes \sum_{b \in O_{B,j}}F_{(y_1,y_2)b}\ket{\psi}= 0
\end{align*}
for $i \neq j$, since $S$ is a perfect quantum strategy for the $(\mathcal{G}_1\lor \mathcal{G}_2)$-game. Thus, we have 
\begin{align}
    p_{(x_1,x_2)i} \varphi (q_{(y_1,y_2)j})^{\tp} \varphi=0\label{orthogonal}
\end{align}
for all $i \neq j$. Using equations \eqref{equal1} and \eqref{orthogonal} several times, we get
\begin{align*}
    p_{(x_1,x_2)i} \varphi^2&= p_{(x_1,x_2)i} \varphi ((q_{(y_1,y_2)i})^{\tp}+(q_{(y_1,y_2)j})^{\tp})\varphi\\
    &=p_{(x_1,x_2)i} \varphi (q_{(y_1,y_2)i})^{\tp}\varphi\\
    &=(p_{(x_1,x_2)i}+p_{(x_1,x_2)j}) \varphi (q_{(y_1,y_2)i})^{\tp}\varphi\\
    &=\varphi (q_{(y_1,y_2)i})^{\tp}\varphi\\
    &=p_{(s_1,s_2)i}\varphi (q_{(y_1,y_2)i})^{\tp}\varphi\\
    &=p_{(s_1,s_2)i}\varphi^2
\end{align*}
for all $(x_1,x_2), (s_1,s_2) \in I_{A,1}\times I_{A,2}$ and $i,j \in \{1,2\}, i \neq j$. Since $\ket{\psi}$ has full Schmidt rank, $\varphi$ is invertible and we obtain $ p_{(x_1,x_2)i}= p_{(s_1,s_2)i}$ for all $(x_1,x_2), (s_1,s_2) \in I_{A,1}\times I_{A,2}$. One similarly shows $q_{(y_1,y_2)j}=q_{(t_1,t_2)j}$. From now on, we let $p_i:=p_{(x_1,x_2)i}$, $q_j:=q_{(y_1,y_2)j}$ for $i=1,2$, $j=1,2$.

We will now show that $p_i$ and $q_j$ are projections. Using equations \eqref{equal1} and \eqref{orthogonal}, we obtain
\begin{align*}
    p_1p_2\varphi^2&=p_1p_2\varphi(q_1^{\tp}+q_2^{\tp})\varphi\\
    &=p_1p_2\varphi(q_2^{\tp})\varphi\\
    &=p_1(p_1+p_2)\varphi(q_2^{\tp})\varphi\\
    &=p_1\varphi(q_2^{\tp})\varphi\\
    &=0.
\end{align*}
Since $\varphi$ is invertible, we obtain $p_1p_2=0$. In particular, we have $p_1=p_1(p_1+p_2)=p_1^2$ and since we already know $p_1=p_1^*$, we get that $p_1$ is a projection. It immediately follows that $p_2=1-p_1$ is a projection. One can similarly show that $q_1$ and $q_2$ are projections.

In the next step, we show $p_1 \neq 0 \neq q_1$. By \eqref{lemmaassumption} and \eqref{same}, we have 
\begin{align*}
    \bra{\psi}E_{(x_1,x_2)a}\otimes F_{(y_1,y_2)b}\ket{\psi}=\bra{\psi'}E_{(x_1,x_2)a}'\otimes F_{(y_1,y_2)b}'\ket{\psi'}> 0
\end{align*}
for some $a \in O_{A,1}$ and $b \in O_{B,1}$. Since all $E_{(x_1,x_2)a}$, $F_{(y_1,y_2)b}$ are positive, it holds
\begin{align*}
    \bra{\psi}p_1 \otimes q_1 \ket{\psi}\geq\bra{\psi}E_{(x_1,x_2)a}\otimes F_{(y_1,y_2)b}\ket{\psi} >0 
\end{align*}
which implies $p_1 \neq 0 \neq q_1$.

In the last step, we construct a perfect quantum strategy for $\mathcal{G}_1$. Fix some $x_2 \in O_{A,2}$ and $y_2 \in O_{B,2}$ and define the quantum strategy $\tilde{S}$ for $\mathcal{G}_1$ as follows.
\begin{align*}
    \ket{\tilde{\psi}}&:=(p_1 \otimes q_1) \ket{\psi}\in p_1\mathcal{H}_A\otimes q_1\mathcal{H}_B,\\
    \tilde{E}_{x_1a}&:=E_{(x_1,x_2)a}, \\
    \tilde{F}_{y_1b}&:=F_{(y_1,y_2)b}
\end{align*}
for $x_1\in I_{A,1}$, $y_1\in I_{B,1}$ and $a \in O_{A,1}$, $b \in O_{B,1}$. Note that $p_1$ is the identity in $B(p_1\mathcal{H}_A)$ and similarly $q_1$ is the identity in $B(q_1\mathcal{H}_B)$. Using Lemma \ref{projeq1}, we see that $\tilde{E}_{x_1a}\in B(p_1\mathcal{H}_A)$ and $\tilde{F}_{y_1b}\in B(q_1\mathcal{H}_B)$. Thus $\{\tilde{E}_{x_1a}\,|\, a \in O_{A,1}\}$ and $\{\tilde{F}_{y_1b}\,|\, b \in O_{B,1}\}$ are POVMs on the space $p_1\mathcal{H}_A\otimes q_1\mathcal{H}_B$ for all $x_1$, $y_1$. 
Furthermore, we have
\begin{align*}
    \bra{\tilde{\psi}} \tilde{E}_{x_1a} \otimes \tilde{F}_{y_1b} \ket{\tilde{\psi}}&=\bra{\psi} p_1E_{(x_1,x_2)a}p_1 \otimes q_1F_{(y_1,y_2)b}q_1 \ket{\psi}\\
    &=\bra{\psi}E_{(x_1,x_2)a} \otimes F_{(y_1,y_2)b} \ket{\psi}
\end{align*}
by Lemma \ref{projeq1}. Since $S=(\ket{\psi}, \{E_{(x_1,x_2)a}\}_{(x_1,x_2)},\{F_{(y_1,y_2)b}\}_{(y_1,y_2)})$ is a perfect quantum strategy for the $(\mathcal{G}_1 \lor \mathcal{G}_2)$-game, we know that 
\begin{align*}
    \bra{\psi}E_{(x_1,x_2)a} \otimes F_{(y_1,y_2)b} \ket{\psi}=0
\end{align*}
whenever $V((x_1,x_2), (y_1, y_2), a,b)=0$. But since we have $a\in O_{A,1}$ and $b\in O_{B,1}$, this is by definition equivalent to $V_{\mathcal{G}_1}(x_1, y_1, a,b)=0$. We conclude that $\tilde{S}$ is a perfect quantum strategy for $\mathcal{G}_1$.
\end{proof}

We believe that the $(\mathcal{G}_1 \lor \mathcal{G}_2)$-game is of independent interest. We now list a few natural directions of further inquiry. The first one that comes to mind is what happens if we consider games with quantum advantage, but not necessarily perfect quantum strategies in Lemma \ref{G1orG2pseudotele} and Lemma \ref{perfectqstratG1}. Furthermore, one might ask about the non-signalling strategies of the $(\mathcal{G}_1 \lor \mathcal{G}_2)$-game. We leave those as open questions here. In this article, we will use the $(\mathcal{G}_1 \lor \mathcal{G}_2)$-game in the next sections to construct non-robust self-tests and games that do not self-test any states.

\section{A non-robust self-test}\label{sect:Nonrobustselftest}

In this section, we construct a game that non-robustly self-tests a perfect quantum strategy. The idea is to consider the $(\mathcal{G}_1\lor \mathcal{G}_2)$-game with the following games. We let $\mathcal{G}_1$ be a game that has no perfect quantum strategy, but a sequence of quantum strategies whose winning probabilities converge to $1$. Note that such a game was constructed by Slofstra in \cite{SlofstraQcor}. For $\mathcal{G}_2$ we take a pseudo-telepathy game that self-tests a quantum strategy. Then the $(\mathcal{G}_1\lor \mathcal{G}_2)$-game still self-tests this strategy, because $\mathcal{G}_1$ has no perfect quantum strategy. This self-test is not robust, however, since we can construct a near-optimal strategies of $(\mathcal{G}_1\lor \mathcal{G}_2)$-game from the ones of $\mathcal{G}_1$ and these strategies are not close to the self-tested strategy.

For our proof technique to go through, we need to choose a game $\mathcal{G}_2$ that is synchronous. In this case we can ensure that if one of the players gets a pair of questions $(x_1,x_2)$ and chooses to play game $\mathcal{G}_1$, the question $x_2$ does not matter for the output (see part $(ii)$ below). 

\begin{lemma}\label{samefordifferentindex}
Let $\mathcal{G}_1$ be a nonlocal game and let $\mathcal{G}_2$ be a synchronous nonlocal game. Furthermore, let $S=(\ket{\psi},\{E_{(x_1,x_2)a}\}_{(x_1,x_2)},\{F_{(y_1,y_2)b}\}_{(y_1,y_2)})$ be a perfect quantum strategy for the $(\mathcal{G}_1 \lor \mathcal{G}_2)$-game. 
\begin{itemize}
    \item[(i)] It holds $(E_{(x_1,x_2)a}\otimes 1)\ket{\psi}=(1 \otimes F_{(y_1, x_2)a})\ket{\psi}$ for all $x_1\in I_{A,1}, y_1\in I_{B,1}$, $x_2 \in I_2$ and $a\in O_2$.
    \item[(ii)] For all  $x_1,x_3\in I_{A,1}$, $x_2 \in I_{2}$, $a\in O_2$ and $y_1,y_3\in I_{B,1}$, $y_2 \in I_{2}$, $b \in O_2$, we have
    \begin{align*}
        (E_{(x_1,x_2)a} \otimes 1) \ket{\psi}=(E_{(x_3,x_2)a} \otimes 1) \ket{\psi} \,\text{ and }\,
        (1\otimes F_{(y_1,y_2)b}) \ket{\psi}=(1 \otimes F_{(y_3,y_2)b}) \ket{\psi}.
\end{align*}
    \item[(iii)] If $\ket{\psi}$ has full Schmidt rank, then $E_{(x_1,x_2)a}=E_{(x_3,x_2)a}$ for $x_1,x_3\in I_{A,1}$, $x_2 \in I_{2}$, $a\in O_2$ and $F_{(y_1,y_2)b}=F_{(y_3,y_2)b}$ for $y_1,y_3\in I_{B,1}$, $y_2 \in I_{2}$, $b \in O_2$.
\end{itemize}
\end{lemma}

\begin{proof}
Let $S$ be as above and $x_1\in I_{A,1}$, $x_2 \in I_{2}$, $a \in O_2$. Since $\mathcal{G}_2$ is synchronous, it holds 
\begin{align*}
    \bra{\psi}E_{(x_1,x_2)a} \otimes 1 \ket{\psi}=\bra{\psi}E_{(x_1,x_2)a} \otimes F_{(y_1,x_2)a} \ket{\psi}=\bra{\psi}1 \otimes F_{(y_1,x_2)a}\ket{\psi}
\end{align*}
for all $y_1 \in I_{B,1}$. Using the equation above and the fact that $E_{(x_1,x_2)a}^2\leq E_{(x_1,x_2)a}$ and $F_{(y_1, x_2)a}^2\leq F_{(y_1, x_2)a}$, we obtain
\begin{align*}
    \|(E_{(x_1,x_2)a}\otimes 1)\ket{\psi}- (1 \otimes F_{(y_1, x_2)a})\ket{\psi}\|^2&=\bra{\psi}E_{(x_1,x_2)a}^2 \otimes 1\ket{\psi}+\bra{\psi}1 \otimes F_{(y_1,x_2)a}^2 \ket{\psi}\\
    &\quad -2\bra{\psi}E_{(x_1,x_2)a} \otimes F_{(y_1,x_2)a} \ket{\psi}\\
    &\leq \bra{\psi}E_{(x_1,x_2)a} \otimes 1\ket{\psi}+\bra{\psi}1 \otimes F_{(y_1,x_2)a}\ket{\psi}\\
    &\quad -2\bra{\psi}E_{(x_1,x_2)a} \otimes F_{(y_1,x_2)a} \ket{\psi}\\
    &=0.
\end{align*}
This yields $(E_{(x_1,x_2)a}\otimes 1)\ket{\psi}=(1 \otimes F_{(y_1, x_2)a})\ket{\psi}$ and thus proves part $(i)$. Since this equation holds for every $x_1\in I_{A,1}$, we get $(E_{(x_1,x_2)a}\otimes 1)\ket{\psi}=(E_{(x_3,x_2)a}\otimes 1)\ket{\psi}$, $x_1, x_3\in I_{A,1}$ and similarly $(1\otimes F_{(y_1,x_2)a}) \ket{\psi}=(1 \otimes F_{(y_3,x_2)a}) \ket{\psi}$ for $y_1,y_3\in I_{B,1}$ which yields $(ii)$. Part $(iii)$ follows from $(ii)$ since $\ket{\psi}\bra{\psi}$ is invertible if $\ket{\psi}$ has full Schmidt rank. 
\end{proof}

We can now present the non-robust self-test. 

\begin{theorem}\label{thm:nonrobustselftest}
Let $\mathcal{G}_1$ be a nonlocal game that has no perfect quantum strategy, but there exists a sequence of quantum strategies whose winning probabilities converge to $1$. Let $\mathcal{G}_2$ be a synchronous pseudo-telepathy game that self-tests a strategy $S_2=(\tilde{\ket{\psi}} \in \tilde{\mathcal{H}}_A\otimes \tilde{\mathcal{H}}_B, \{\hat{E}_{x_2a}\}_{x_2},\{\hat{F}_{y_2b}\}_{y_2})$. In this case the $(\mathcal{G}_1 \lor \mathcal{G}_2)$-game is a non-robust self-test for the strategy $\tilde{S}_2=(\tilde{\ket{\psi}},\{\tilde{E}_{(x_1,x_2)a}\}_{(x_1,x_2)},$ $\{\tilde{F}_{(y_1,y_2)b}\}_{(y_1,y_2)})$, where
\begin{align*}
    \tilde{E}_{(x_1,x_2)a}=\begin{cases}\hat{E}_{x_2a} \text{ for } a \in O_{A,2},\\ 0 \text{ otherwise,}\end{cases}
    \tilde{F}_{(y_1,y_2)b}=\begin{cases}\hat{F}_{y_2b} \text{ for } b \in O_{B,2},\\ 0 \text{ otherwise.}\end{cases}
\end{align*}
\end{theorem}

\begin{proof}
Let $S'=(\ket{\psi'},\{E'_{(x_1,x_2)a}\}_{(x_1,x_2)},\{F'_{(y_1,y_2)b}\}_{(y_1,y_2)})$ be a perfect quantum strategy for the $(\mathcal{G}_1\lor \mathcal{G}_2)$-game. To show that the $(\mathcal{G}_1 \lor \mathcal{G}_2)$-game is a self-test for the strategy $\tilde{S}_2$, we have to prove that $\tilde{S}_2$ is a local dilation of $S'$.

By assumption, the game $\mathcal{G}_1$ has no perfect quantum strategy, thus we get $\bra{\psi'}E_{(x_1,x_2)a}'\otimes F_{(y_1,y_2)b}' \ket{\psi'}=0$ for all $a \in O_{A,1}$, $b \in O_{B,1}$ by Lemma \ref{perfectqstratG1}.  Summing over all $b \in O_{B,1}\dot\cup O_{B,2}$ and $a \in O_{A,1}\dot\cup O_{A,2}$, respectively, yields $\bra{\psi'}E_{(x_1,x_2)a}'\otimes 1 \ket{\psi'}=0$, $\bra{\psi'}1\otimes F_{(y_1,y_2)b}' \ket{\psi'}=0$ for $a \in O_{A,1}$, $b \in O_{B,1}$. Therefore, using $(E'_{(x_1,x_2)a})^2\leq E'_{(x_1,x_2)a}$, we have 
\begin{align*}
    \|(E'_{(x_1,x_2)a}\otimes 1)\ket{\psi'}\|^2&=\bra{\psi'}(E'_{(x_1,x_2)a})^2\otimes 1)\ket{\psi'}\\
    &\leq \bra{\psi'}E'_{(x_1,x_2)a}\otimes 1)\ket{\psi'}\\
    &=0. 
\end{align*}
We conclude $(E'_{(x_1,x_2)a}\otimes 1)\ket{\psi'}=0$ for all $a\in O_{A,1}$. We similarly get $(1 \otimes F'_{(y_1,y_2)b})\ket{\psi'}=0$ for all $b\in O_{B,1}$. This especially yields that $\mathrm{supp}_A(\ket{\psi'})$ and $\mathrm{supp}_B(\ket{\psi'})$ are invariant under those $E_{(x_1,x_2)a}'$  and $F_{(y_1,y_2)b}'$, respectively. By Lemma \ref{samefordifferentindex} $(i)$, we know $(E'_{(x_1,x_2)a}\otimes 1)\ket{\psi'}=(1 \otimes F'_{(y_1, x_2)a})\ket{\psi'}$ for all $x_1\in I_{A,1}, y_1\in I_{B,1}$, $x_2 \in I_2$ and $a\in O_2$. Therefore, Lemma \ref{lem:invariance} yields that $\mathrm{supp}_A(\psi')$ is invariant under $E_{(x_1,x_2)a}'$ and $\mathrm{supp}_B(\psi')$ is invariant under $F_{(y_1,y_2)b}'$ in those cases. We conclude that $\mathrm{supp}_A(\ket{\psi'})$ and $\mathrm{supp}_B(\ket{\psi'})$ are invariant under all $E_{(x_1,x_2)a}'$  and $F_{(y_1,y_2)b}'$, respectively. Lemma \ref{lem:DilationfullSchmidtrank} now yields that there exists a perfect quantum strategy  $S=(\ket{\psi},\{E_{(x_1,x_2)a}\}_{(x_1,x_2)},\{F_{(y_1,y_2)b}\}_{(y_1,y_2)})$ such that $S$ is a local dilation of $S'$ and $\ket{\psi}$ has full Schmidt rank, \emph{i.e.} $\ket{\psi}=\sum_{i=1}^d\lambda_i e_i \otimes e_i$ with $\lambda_i >0$ for all $i \in [d]$.

Let $\varphi=\mathrm{diag}(\lambda_i)$, where $\lambda_i$ are the Schmidt coefficients of $\ket{\psi}$. Once again, by Lemma \ref{perfectqstratG1}, we have $\bra{\psi}E_{(x_1,x_2)a}\otimes F_{(y_1,y_2)b} \ket{\psi}=0$ for all $a \in O_{A,1}$, $b \in O_{B,1}$. Summing over all $b \in O_{B,1}\dot\cup O_{B,2}$ and $a \in O_{A,1}\dot\cup O_{A,2}$, respectively, yields
\begin{align*}
    \mathrm{Tr}(E_{(x_1,x_2)a}\varphi^2)&=\bra{\psi}E_{(x_1,x_2)a}\otimes 1 \ket{\psi}=0,\\
    \mathrm{Tr}((\varphi^*)^2 (F_{(y_1,y_2)b})^{\tp})&=\bra{\psi}1\otimes F_{(y_1,y_2)b} \ket{\psi}=0
\end{align*}
for $a \in O_{A,1}$, $b \in O_{B,1}$. Since $\ket{\psi}$ has full Schmidt rank, we obtain $E_{(x_1,x_2)a}=0$ and $F_{(y_1,y_2)b}=0$ for all $a \in O_{A,1}$, $b \in O_{B,1}$. Therefore, $\{E_{(x_1,x_2)a}\,|\, a \in O_{A,2}\}$ and $\{F_{(y_1,y_2)b}\,|\, b \in O_{B,2}\}$ are POVMs. For fixed $x_1 \in I_{A,1}$, $y_1 \in I_{B,1}$, we get that $(\ket{\psi},\{E_{(x_1,x_2)a}\}_{(x_1,x_2)},\allowbreak\{F_{(y_1,y_2)b}\}_{(y_1,y_2)})$ is a perfect quantum strategy for $\mathcal{G}_2$.
Because $\mathcal{G}_2$ is a self-test for $S_2$, we know that for fixed $x_1,y_1$ there exists
Hilbert spaces $\mathcal{H}_{A,aux}$ and $\mathcal{H}_{B,aux}$, a state $\ket{aux}\in \mathcal{H}_{A,aux}\otimes\mathcal{H}_{B,aux}$ and isometries $U_A: \mathcal{H}_A \to \tilde{\mathcal{H}}_A \otimes \mathcal{H}_{A,aux}$, $U_B:\mathcal{H}_B \to \tilde{\mathcal{H}}_B \otimes \mathcal{H}_{B,aux}$ such that with $U:=U_A \otimes U_B$ it holds
\begin{align*}
    U\ket{\psi}&= \ket{\tilde{\psi}}\otimes \ket{aux},\\
    U(E_{(x_1,x_2)a}\otimes 1)\ket{\psi}&=[(\hat{E}_{x_2a}\otimes 1)\tilde{\ket{\psi}}]\otimes \ket{aux},\\
    U(1\otimes F_{(y_1,y_2)b} )\ket{\psi}&=[(1\otimes \hat{F}_{(y_1,y_2)b})\tilde{\ket{\psi}}]\otimes \ket{aux}.
\end{align*}
Note that from Lemma \ref{samefordifferentindex} $(iii)$, we have $E_{(x_1,x_2)a}=E_{(x_3,x_2)a}$ for $x_1,x_3\in I_{A,1}$, $x_2 \in I_{2}$, $a\in O_2$ and $F_{(y_1,y_2)b}=F_{(y_3,y_2)b}$ for $y_1,y_3\in I_{B,1}$, $y_2 \in I_{2}$, $b \in O_2$. Thus, we get 
\begin{align*}
    U(E_{(x_3,x_2)a}\otimes 1)\ket{\psi}&=U(E_{(x_1,x_2)a}\otimes 1)\ket{\psi}\\
    &=[(\hat{E}_{x_2a}\otimes 1)\tilde{\ket{\psi}}]\otimes \ket{aux}\\
    &=[(\tilde{E}_{(x_3,x_2)a}\otimes 1)\tilde{\ket{\psi}}]\otimes \ket{aux}
\end{align*}
for all $(x_3,x_2)\in I_{A,1}\times I_{A,2}$, $a \in O_2$. We similarly get 
 $U(1\otimes F_{(y_3,y_2)b} )\ket{\psi}=[(1\otimes \tilde{F}_{(y_3,y_2)b})\tilde{\ket{\psi}}]\otimes \ket{aux}$ for all $(y_3,y_2)\in I_{B,1}\times I_{B,2}$, $a \in O_2$. Since we know $E_{(x_3,x_2)a}=0$ and $F_{(y_3,y_2)b}=0$ for all $a \in O_{1}$, $b \in O_{1}$, we deduce that $\tilde{S_2}$ is a local dilation of $S$.
Summarizing, we have that $\tilde{S_2}$ is a local dilation of $S$ and $S$ is a local dilation of $S'$. By Lemma \ref{transitivity}, we get that $\tilde{S_2}$ is a local dilation of $S'$, thus the $(\mathcal{G}_1 \lor \mathcal{G}_2)$-game is a self-test for $\tilde{S}_2$. 
 
 It remains to show that this self-test is not robust. Since there exists a sequence of quantum strategies for $\mathcal{G}_1$ whose winning probability converges to $1$, for every $\delta>0$ there is a quantum strategy $\hat{S}_{\delta}=(\ket{\psi^{(\delta)}},\{\hat{E}^{(\delta)}_{x_1a}\}_{x_1},\{\hat{F}^{(\delta)}_{y_1b}\}_{y_1}))$ with winning probability at least $1-\delta$. By defining
 \begin{align*}
    E^{(\delta)}_{(x_1,x_2)a}:=\begin{cases}\hat{E}^{(\delta)}_{x_1a} \text{ for } a \in O_{A,1},\\ 0 \text{ otherwise,}\end{cases}
    F^{(\delta)}_{(y_1,y_2)b}:=\begin{cases}\hat{F}^{(\delta)}_{y_1b} \text{ for } b \in O_{B,1},\\ 0 \text{ otherwise,}\end{cases}
\end{align*}
we obtain a strategy $S_{\delta}=(\ket{\psi^{(\delta)}},\{E^{(\delta)}_{(x_1,x_2)a}\}_{(x_1,x_2)},\{F^{(\delta)}_{(y_1,y_2)b}\}_{(y_1,y_2)}))$ for the $(\mathcal{G}_1 \lor \mathcal{G}_2)$-game with the same winning probability as $\hat{S}^{(\delta)}$ (\emph{i.e.} at least $1-\delta$). Since we have $E^{(\delta)}_{(x_1,x_2)a}=0$ for all $a \in O_{A,2}$ and all $\delta>0$, we see that
\begin{align*}
    \|U(E^{(\delta)}_{(x_1,x_2)a} \otimes 1)\ket{\psi^{(\delta)}}-[(\tilde{E}_{(x_1,x_2)a}\otimes 1)\tilde{\ket{\psi}}]\otimes \ket{aux'} \|=\|[(\tilde{E}_{(x_1,x_2)a}\otimes 1)\tilde{\ket{\psi}}]\otimes \ket{aux'}\|
\end{align*}
for all $a \in O_{A,2}$, all $\delta>0$, all suitable isometries $U$ and all auxiliary states $\ket{aux'}$. Since we have $\tilde{E}_{(x_1,x_2)a}=0$ for all $a \in O_{A,1}$, we know that there is $a_0 \in O_{A,2}$ such that $\tilde{E}_{(x_1,x_2)a_0}\neq 0$ and thus 
\begin{align*}
    \|[(\tilde{E}_{(x_1,x_2)a_0}\otimes 1)\tilde{\ket{\psi}}]\otimes \ket{aux'}\|>\varepsilon'
\end{align*}
for some $\varepsilon'>0$. Summarizing, we found an $\varepsilon'>0$ such that for all $\delta>0$, we have a $\delta$-optimal strategy $S^{(\delta)}$ such that 
\begin{align*}
       \|U(E^{(\delta)}_{(x_1,x_2)a_0} \otimes 1)\ket{\psi^{(\delta)}}-[(\tilde{E}_{(x_1,x_2)a_0}\otimes 1)\tilde{\ket{\psi}}]\otimes \ket{aux'} \|>\varepsilon'
\end{align*}
for all isometries $U$. This shows that the $(\mathcal{G}_1\lor \mathcal{G}_2)$-game is a non-robust self-test for $\tilde{S}_2$.
\end{proof}

\begin{example}\label{ex:nonrobust}
By \cite{SlofstraQcor}, there exists a linear constraint system game that has no perfect quantum strategy, but a sequence of strategies whose winning probabilities converge to $1$. The proof is constructive, the linear system has $184$ equations and $235$ variables. Let $\mathcal{G}_1$ be this linear constraint system game. We have $|I_{A,1}|=184$, $|I_{B,1}|=235$ and $|O_{A,1}|=8$, $|O_{B,1}|=2$. We let $\mathcal{G}_2$ be the synchronous version of the magic square game (see Subsection \ref{syncmagicsquare}), for which we know that it is synchronous pseudo-telepathy game. We have $|I_{A,2}|=|I_{B,2}|=6$ and $|O_{A,2}|=|O_{B,2}|=8$. By Corollary \ref{cor:Selftestsyncmagicsquare}, we know that it self-tests the perfect quantum strategy $S_2$. Theorem \ref{thm:nonrobustselftest} shows that the $(\mathcal{G}_1\lor \mathcal{G}_2)$-game is a non-robust self-test for the strategy $\tilde{S}_2$. For this game, we have $|I_{A}|=1104$, $|I_B|=1410$ and $|O_{A}|=16$, $|O_B|=10$. 
\end{example}

\section{Games that do not self-test states}\label{sect:gamesnotselfteststates}

We will now construct games that do not self-test any state. We will once more use the $(\mathcal{G}_1\lor \mathcal{G}_2)$-game. We first show that a game does not self-test any state if it has two optimal strategies using states of coprime Schmidt rank. 

\begin{lemma}\label{nonselftest}
Let $\mathcal{G}$ be a nonlocal game such that $\omega^*(\mathcal{G})>\omega(\mathcal{G})$. Let 
\[S_1=\left(\ket{\psi_1}\in \mathcal{H}_A^{(1)}\otimes \mathcal{H}_B^{(1)} , \{E_{xa}^{(1)}\}_x, \{F_{yb}^{(1)}\}_y\right), \quad
S_2=\left(\ket{\psi_2} \in \mathcal{H}_A^{(2)}\otimes \mathcal{H}_B^{(2)}, \{E_{xa}^{(2)}\}_x, \{F_{yb}^{(2)}\}_y\right)\]
be two optimal quantum strategies. If the Schmidt ranks of $\ket{\psi_1}$ and $\ket{\psi_2}$ are coprime (i.e. $\mathrm{gcd}(n_1, n_2)=1$ for Schmidt ranks $n_1, n_2$), then $\mathcal{G}$ does not self-test any state $\ket{\tilde{\psi}}$ of an optimal quantum strategy.
\end{lemma}

\begin{proof}
We prove this lemma by contradiction. Assume that $\mathcal{G}$ self-tests the state $\ket{\tilde{\psi}}\in \tilde{\mathcal{H}}_A\otimes \tilde{\mathcal{H}}_B$. The state $\ket{\tilde{\psi}}$ has Schmidt rank $d>1$ as otherwise the classical value and the quantum value of $\mathcal{G}$ coincide. Let $n_1, n_2$ be the Schmidt ranks of the states $\ket{\psi_1}\in \mathcal{H}_A^{(1)}\otimes \mathcal{H}_B^{(1)}$ and $\ket{\psi_2}\in \mathcal{H}_A^{(2)}\otimes \mathcal{H}_B^{(2)}$. Since $\mathcal{G}$ self-tests $\ket{\tilde{\psi}}$, we get isometries $U_A^{(i)}: \mathcal{H}_A^{(i)} \to \tilde{\mathcal{H}}_A \otimes \mathcal{H}_{A,aux}^{(i)}$, $U_B:\mathcal{H}_B^{(i)} \to \tilde{\mathcal{H}}_B \otimes \mathcal{H}_{B,aux}^{(i)}$ and states $\ket{aux_i}\in \mathcal{H}_{A,aux}^{(i)}\otimes\mathcal{H}_{B,aux}^{(i)}$, $i=1,2$, such that with $U_i=U_A^{(i)}\otimes U_B^{(i)}$ it holds that
\begin{align*}
    U_1\ket{\psi_1}= \ket{\tilde{\psi}}\otimes \ket{aux_1}
    \intertext{and}
    U_2\ket{\psi_2}= \ket{\tilde{\psi}}\otimes \ket{aux_2}.
\end{align*}
Note that since we have $U_i=U_A^{(i)}\otimes U_B^{(i)}$, the states $U_i\ket{\psi_i}$ have the same the Schmidt rank with respect to the bipartition $\tilde{\mathcal{H}}_A \otimes \mathcal{H}_{A,aux}^{(i)}$, $\tilde{\mathcal{H}}_B \otimes \mathcal{H}_{B,aux}^{(i)}$ as $\ket{\psi_i}$ with respect to $\mathcal{H}_A^{(i)}$, $\mathcal{H}_B^{(i)}$. Furthermore, the Schmidt rank of $\ket{\tilde{\psi}}\otimes \ket{aux_i}$ with respect to the bipartition $\tilde{\mathcal{H}}_A \otimes \mathcal{H}_{A,aux}^{(i)}$, $\tilde{\mathcal{H}}_B \otimes \mathcal{H}_{B,aux}^{(i)}$ is equal to $d c_i$ for $c_i$ being the Schmidt rank of $\ket{aux_i}$. Thus, comparing the Schmidt ranks with respect to $\tilde{\mathcal{H}}_A \otimes \mathcal{H}_{A,aux}^{(i)}$, $\tilde{\mathcal{H}}_B \otimes \mathcal{H}_{B,aux}^{(i)}$ in the equations above yields $n_1=d c_1$ and $n_2=d c_2$. This shows that $d>1$ is a common divisor of $n_1$ and $n_2$ which contradicts our assumption that $n_1$ and $n_2$ are coprime.
\end{proof}

The next theorem shows that the $(\mathcal{G}_1\lor \mathcal{G}_2)$-game has two perfect quantum strategies using states with coprime Schmidt rank if the games $\mathcal{G}_1$ and $\mathcal{G}_2$ have perfect quantum strategies with states of coprime Schmidt rank. Thus we are just left with finding two nonlocal games whose perfect quantum strategies have states with coprime Schmidt rank. 

\begin{theorem}\label{thm:2pseud}
Let $\mathcal{G}_1$ and $\mathcal{G}_2$ be pseudo-telepathy games with perfect quantum strategies 
\[S_1=\left(\ket{\psi_1}\in \mathcal{H}_A^{(1)}\otimes \mathcal{H}_B^{(1)} , \{E_{xa}^{(1)}\}_x, \{F_{yb}^{(1)}\}_y\right),
S_2=\left(\ket{\psi_2} \in \mathcal{H}_A^{(2)}\otimes \mathcal{H}_B^{(2)}, \{E_{xa}^{(2)}\}_x, \{F_{yb}^{(2)}\}_y\right),\]
respectively. Assume that $\ket{\psi_1}$ and $\ket{\psi_2}$ have coprime Schmidt ranks. Then the $(\mathcal{G}_1\lor \mathcal{G}_2)$-game does not self-test any state. 
\end{theorem}

\begin{proof}
We get perfect quantum strategies $\tilde{S_1}=(\ket{\psi_1}, \{\tilde{E}_{(x_1,x_2)a}^{(1)}\}_{(x_1,x_2)},\allowbreak \{\tilde{F}_{(y_1,y_2)b}^{(1)}\}_{(y_1,y_2)})$ and $\tilde{S_2}=(\ket{\psi_2}, \{\tilde{E}_{(x_1,x_2)a}^{(2)}\}_{(x_1,x_2)}, \{\tilde{F}_{(y_1,y_2)b}^{(2)}\}_{(y_1,y_2)})$ for the $(\mathcal{G}_1\lor \mathcal{G}_2)$-game by defining 
\begin{align*}
    \tilde{E}_{(x_1,x_2)a}^{(i)}:=\begin{cases}E_{x_i a} \text{ for } a \in O_{A,i},\\0 \text{ otherwise,}\end{cases}
    \tilde{F}_{(y_1,y_2)b}^{(i)}:=\begin{cases}F_{y_i a} \text{ for } a \in O_{B,i},\\0 \text{ otherwise,}\end{cases}
\end{align*}
for $i=1,2$. Since $\ket{\psi_1}$ and $\ket{\psi_2}$ have coprime Schmidt rank, Lemma \ref{nonselftest} shows that the $(\mathcal{G}_1\lor \mathcal{G}_2)$-game does not self-test any state. 
\end{proof}

In the following, we will construct an explicit game fulfilling the conditions of Theorem \ref{thm:2pseud}. Note that if we look at strategies involving maximally entangled states on $\C^{d_i\times d_i}$, then it suffices to find two such strategies with coprime $d_1$ and $d_2$. In our case, we will have $d_1=3$, $d_2=4$. We can get a perfect quantum strategy with $d_2=4$ from the magic square game (see Subsection \ref{syncmagicsquare}). For $d_1=3$, we will use an independent set game with a graph coming from a $3$-dimensional weak Kochen-Specker set. This will be explained in the next subsections. 

\subsection{Independent set game}\label{sect:indepset}

In this subsection we discuss the independent set game, which was introduced in \cite{Qhom}. We will see the connection between quantum independent sets and perfect quantum strategies for the game. For us, a graph $G$ is always finite, simple and undirected. Thus, it consists of a finite vertex set $V(G)$ and an edge set $E(G)$ which is a set of unordered pairs of vertices.  

\begin{definition}
Let $G$ be a graph. An independent set of size $t$ in the graph $G$ is a set of vertices $\{v_1, \dots, v_t\} \in V(G)$ such that $(v_i, v_j)\notin E(G)$ for all $i\neq j$. The \emph{independence number} $\alpha(G)$ denotes the size of the largest independent set in $G$. 
\end{definition}

For a natural number $t\in \N$ and a graph $G$, the $(G,t)$-\emph{independent set game} is played with two players Alice and Bob, and a referee. Alice and Bob try to convince the referee that they know an independent set of size $t$ of the graph $G$. The game is played as follows. The referee sends the players natural numbers $x_A, x_B \in [t]$ and the players answer with vertices $v_A, v_B \in V(G)$. In order to win the $(G,t)$-independent set game, the following conditions must be met:
\begin{itemize}
\item[(1)] If $x_A=x_B$, then $v_A=v_B$,
\item[(2)] If $x_A\neq x_B$, then $v_A \neq v_B$ and $(v_A, v_B) \notin E(G)$.
\end{itemize}
The players can agree on a strategy beforehand, but cannot communicate during the game. 

Note that the $(G,t)$-independent set game can be thought of as the $(K_t, \Bar{G})$ - homomorphism game, see \cite{Qhom}.

\begin{definition}\label{defqindepset}
Let $G$ be a graph. A \emph{quantum independent set} of size $t$ in $G$ is a collection $P=\{P_{xu}\}_{x\in [t],u \in V(G)}$ of projections $P_{xu}\in \C^{d\times d}$ such that
\begin{itemize}
    \item[(i)] $\sum_{u\in V(G)} P_{xu}= 1_{\C^{d\times d}}$ for all $x\in [t]$,
    \item[(ii)] $P_{xu}P_{yv}=0$ for $(u,v)\in E(G)$ and all $x,y\in [t]$,
    \item[(iii)] $P_{xu}P_{yu}=0$ for all $x\neq y$, $u\in V(G)$.
\end{itemize}
The {quantum independence number} $\alpha_q(G)$ denotes the maximum number $t$ such that there exists a quantum independent set of size $t$ in $G$.
\end{definition}

The following lemma follows from \cite[Section 2.1\& 2.2]{Qhom}.

\begin{lemma}\label{specificstrat}
Let $G$ be a graph and let $P=\{P_{xu}\}_{x\in [t],u \in V(G)}$ be a quantum independent set of size $t$ in $G$ of projections $P_{xu}\in \C^{d\times d}$. 
\begin{itemize}
    \item[(i)] The strategy $(\frac{1}{\sqrt{d}}\sum_{i=1}^d e_i \otimes e_i, \{P_{xu}\}_x, \{(P_{yv})^{\tp}\}_y)$ is a perfect quantum strategy of the $(G,t)$-independent set game.
    \item[(ii)] If $t> \alpha(G)$, then the $(G,t)$-independent set game is a pseudo-telepathy game.\label{pseudo}
\end{itemize}
\end{lemma}

By the previous lemma, we see that quantum independent sets yield perfect quantum strategies with maximally entangled state for the independent set game. Thus, our goal is to construct a graph that has quantum independent sets with projections in odd dimension. This will be done in the next subsection using odd-dimensional Kochen-Specker sets. 

\subsection{Kochen-Specker sets}\label{sect:KS}

To get a counterexample for state self-testing, we use Kochen-Specker sets to construct an explicit independent set game having a perfect quantum strategy with a state of Schmidt rank $3$. Kochen-Specker sets are sets of vectors that provide proofs of the (Bell-)Kochen-Specker theorem \cite{Bell,KS}.
Let $S\subseteq \C^n$ be a set of vectors. A function $f:S\to \{0,1\}$ is a \emph{marking function} for $S$ if for all orthonormal bases $B\subseteq S$, we have $\sum_{v\in B}f(v)=1$.

\begin{definition}
Let $S\subseteq \C^n$ be a set of unit vectors. 
\begin{itemize}
    \item[(i)] The set $S$ is a \emph{Kochen-Specker set} if there is no marking function for $S$.
    \item[(ii)] The set $S$ is a \emph{weak Kochen-Specker set} \cite{RW} if for all marking functions $f$ for $S$ there exist orthogonal vectors $u,v \in S$ such that $f(u)=f(v)=1$.
\end{itemize}
\end{definition}

The above notions were generalized in \cite{MSS} to sets of projections. Let $\mathcal{Q}_n\subseteq \C^{n\times n}$ be the set of all $n \times n$ projections. 
A marking function $f$ for $S \subsetneq \mathcal{Q}_n$ is a function $f:S \to \{0,1\}$ such that for all $M \subseteq S$ with $\sum_{p\in M} p=1_{\C^{n\times n}}$, we have $\sum_{p \in M}f(p)=1$.

\begin{definition}\cite{MSS}
A set $S \subsetneq \mathcal{Q}_n$ is a \emph{projective Kochen-Specker set} if for all marking functions $f$ for $S$, there exists $p,p'\in S$ for which $pp'=0$ and $f(p)=f(p')=1$.
\end{definition}

Let $\{v_1, \dots, v_k\}\subseteq \C^{n}$ be a weak Kochen-Specker set. Then we get a projective Kochen-Specker set by considering the rank one projections $\{v_1v_1\ct, \dots, v_kv_k\ct\} \subseteq \C^{n\times n}$.

\begin{definition}\label{orthogonalitygraph}
Let $S$ be a projective Kochen-Specker set. Let $S_1=\{p_{11}, \dots, p_{1i_1}\}$, \dots, $S_k=\{p_{k1}, \dots, p_{ki_k}\}$ be all subsets of $S$ such that $\sum_{b \in[i_a]}p_{ab}=1$. We define the graph $G_S$ as follows: Let $V(G_S)=\{(a,b)\,|\, a\in [k], b \in [i_a]\}$, where $((a,b),(c,d))\in E(G_S)$ if and only $p_{ab}p_{cd}=0$. Note that $G_S$ is the orthogonality graph of the multiset $S_1 \dot\cup \dots \dot\cup S_k$.
\end{definition}

\begin{remark}
Note that we may have $p_{ab}=p_{cd}$ for some $a,b,c,d$ in the above definition. 
\end{remark}

The next lemma shows that given a projective Kochen-Specker set $S$ we can construct a quantum independent set in the orthogonality graph $G_S$.

\begin{lemma}\label{explicitqindepset}
Let $S \subsetneq \mathcal{Q}_n$ be a projective Kochen-Specker set. Let $S_1=\{p_{11}, \dots, p_{1i_1}\}$, \dots, $S_k=\{p_{k1}, \dots, p_{ki_k}\}$ be all subsets of $S$ such that $\sum_{b \in[i_a]}p_{ab}=1$ and $G_S$ as in Definition \ref{orthogonalitygraph}. Then the collection $Q=\{Q_{j(a,b)}\}_{j\in [k],(a,b)\in V(G_S)}$ with $Q_{j(a,b)}:=\delta_{aj}p_{ab}$ is a quantum independent set of size $k$ of $G_S$.
\end{lemma}

\begin{proof}
We check conditions (i)-(iii) of Definition \ref{defqindepset}. For (i), we compute 
\begin{align*}
    \sum_{(a,b)\in V(G_S)}Q_{j(a,b)}=\sum_{b \in [i_j]} p_{jb}=1
\end{align*}
by definition of $Q$ and choice of $S_j$. Condition (ii) is fulfilled since for $((a,b),(c,d)) \in E(G_S)$, we know $Q_{j(a,b)}Q_{l(c,d)}=\delta_{aj}\delta_{cl}p_{ab}p_{cd}=0$, since $p_{ab}p_{cd}=0$ by definition of $G_S$. For (iii), we have $Q_{j(a,b)}Q_{l(a,b)}=\delta_{aj}\delta_{al}p_{ab}=0$ for $j\neq l$, since $\delta_{aj}\delta_{al}=0$ for $j\neq l$.
\end{proof}

From previous work it is already known that the size of the quantum independent set $Q$ from Lemma \ref{explicitqindepset} is larger than the independence number of the orthogonality graph. 

\begin{theorem}\cite[Theorem 3.4.4]{Scarpa}\label{biggerthanindep}
Let $S$ be a projective Kochen-Specker set and let $S_1$, \dots, $S_k$ and $G_S$ be as in Definition \ref{orthogonalitygraph}. Then $k > \alpha(G_S)$.
\end{theorem}

Combining Lemma \ref{specificstrat}, Lemma \ref{explicitqindepset} and Theorem \ref{biggerthanindep}, we have the following corollary. 

\begin{cor}\label{cor:indepks}
Let $S\subsetneq \mathcal{Q}_n$ be a projective Kochen-Specker set and $G_S$ as in Definition \ref{orthogonalitygraph}. Then the $(G_S, k)$-independent set game is a pseudo-telepathy game with a perfect quantum strategy using the maximally entangled state $\ket{\psi_n}=\frac{1}{\sqrt{n}}\sum_{i=1}^ne_i \otimes e_i$.
\end{cor}

We will now give an example of a $(\mathcal{G}_1\lor \mathcal{G}_2)$-game that does not self-test any state. Note that a $3$-dimensional weak Kochen-Specker set yields a projective Kochen-Specker set $S \subsetneq \mathcal{Q}_3$ and thus the $(G_S,k)$-independent set game has a quantum strategy using the state $\ket{\psi_3}$. 

\begin{example}\label{ex:nostateselftest}
Consider Peres' $3$-dimensional weak Kochen-Specker set $S_1$ with $33$ vectors forming $16$ bases \cite{Peres33}. We get a projective Kochen-Specker set by considering the associated rank-$1$ projections. The orthogonality graph $G_{S_1}$ has $48$ vertices. Furthermore, we get $\alpha(G_{S_1})=15, \alpha_q(G_{S_1})\geq 16$ by using Sage \cite{sagemath} to compute $\alpha(G_{S_i})$ and Lemma \ref{explicitqindepset} for a lower bound on $\alpha_q(G_{S_i})$. By Corollary \ref{cor:indepks}, we know that the $(G_{S_1},16)$-independent set game has a perfect quantum strategy with a state of Schmidt rank $3$. We let $\mathcal{G}_1$ be the $(G_{S_1},16)$-independent set game and let $\mathcal{G}_2$ be the magic square game considered in Subsection \ref{syncmagicsquare}. We know that $\mathcal{G}_2$ has a perfect quantum strategy with a state of Schmidt rank $4$ from Theorem \ref{thm:projselftestms}. Therefore, the $(\mathcal{G}_1\lor \mathcal{G}_2)$-game does not self-test any state by Theorem \ref{thm:2pseud}. In this case, we have $|I_A|=|I_B|=48$, $|O_A|=|O_B|=52$.
\end{example}

\paragraph{Acknowledgments.}\phantom{a}\newline
The authors thank the referees for useful comments and suggestions. S.S. has received funding from the European Union's Horizon 2020 research and innovation programme under the Marie Sklodowska-Curie grant agreement No. 101030346. 
L.M. is funded by the European Union under the Grant Agreement No 101078107, QInteract and VILLUM FONDEN via the QMATH Centre of Excellence (Grant No 10059) and Villum Young Investigator grant (No 37532). This project was funded within the QuantERA II Programme that has received funding from the European Union’s Horizon 2020 research and innovation programme under Grant Agreement No 101017733.

\bibliographystyle{plainurl}
\bibliography{noselftest}

\end{document}